\documentclass[a4paper,11pt]{article}
\usepackage[margin=1.2in]{geometry}

\usepackage[utf8]{inputenc}
\usepackage[T1]{fontenc}

\usepackage{amsmath, amsthm, amssymb}
\usepackage{graphicx}
\usepackage{hyperref}
\hypersetup{
  colorlinks = true,
  citecolor = [rgb]{0.605, 0.257, 0.610} 
}
\usepackage{float}

\usepackage[shortlabels]{enumitem}

\usepackage{aliascnt}
\newtheorem{theorem}{Theorem}[section]

\newaliascnt{lemma}{theorem}
\newtheorem{lemma}[lemma]{Lemma}
\aliascntresetthe{lemma}

\newaliascnt{proposition}{theorem}
\newtheorem{proposition}[proposition]{Proposition}
\aliascntresetthe{proposition}

\newaliascnt{definition}{theorem}
\newtheorem{definition}[definition]{Definition}
\aliascntresetthe{definition}

\newaliascnt{corollary}{theorem}

\aliascntresetthe{corollary}

\newaliascnt{conjecture}{theorem}

\aliascntresetthe{conjecture}

\newaliascnt{claim}{theorem}
\newtheorem{claim}[claim]{Claim}
\aliascntresetthe{claim}

\newaliascnt{observation}{theorem}
\newtheorem{observation}[observation]{Observation}
\aliascntresetthe{observation}

\newaliascnt{remark}{theorem}

\aliascntresetthe{remark}

\renewenvironment{abstract}
{\small
\vspace{-1em}
\begin{center}
\bfseries \abstractname\vspace{-.5em}\vspace{0pt}
\end{center}
\list{}{
  \setlength{\leftmargin}{0.5in}%
  \setlength{\rightmargin}{\leftmargin}%
}%
\item\relax}
{\endlist}

\newcommand\cqedsymbol{\ifmmode$\lrcorner$\else{\unskip\nobreak\hfil
\penalty50\hskip1em\null\nobreak\hfil$\lrcorner$
\parfillskip=0pt\finalhyphendemerits=0\endgraf}\fi} 
\newcommand{\cqed}{\renewcommand{\qed}{\cqedsymbol}}

\usepackage{tabularx}
\usepackage{environ}
\makeatletter
\newcommand{\problemtitle}[1]{\gdef\@problemtitle{#1}}
\newcommand{\probleminput}[1]{\gdef\@probleminput{#1}}
\newcommand{\problemquestion}[1]{\gdef\@problemquestion{#1}}
\NewEnviron{problemgen}{
  \problemtitle{}\probleminput{}\problemquestion{}
  \BODY
  \par\addvspace{.5\baselineskip}
  \noindent
  \begin{tabularx}{\textwidth}{@{\hspace{\parindent}} l X c}
    \multicolumn{2}{@{\hspace{\parindent}}l}{\@problemtitle} \\
    \textbf{Input:} & \@probleminput \\
    \textbf{Output:} & \@problemquestion
  \end{tabularx}
  \par\addvspace{.5\baselineskip}
}
\NewEnviron{decproblem}{
  \problemtitle{}\probleminput{}\problemquestion{}
  \BODY
  \par\addvspace{.5\baselineskip}
  \noindent
  \begin{tabularx}{\textwidth}{@{\hspace{\parindent}} l X c}
    \multicolumn{2}{@{\hspace{\parindent}}l}{\@problemtitle} \\
    \textbf{Input:} & \@probleminput \\
    \textbf{Question:} & \@problemquestion
  \end{tabularx}
  \par\addvspace{.5\baselineskip}
}

\usepackage{xcolor,colortbl}
\usepackage{graphicx}
\usepackage{array}
\usepackage{multirow}
\newcolumntype{P}[1]{>{\centering\arraybackslash}p{#1}}
\newcolumntype{C}[1]{>{\centering\arraybackslash}m{#1}}
\newcolumntype{N}{@{}m{0pt}@{}}
\usepackage{makecell}

\newcommand\eqdef{\overset{\text{def}}{=}}

\def\B{\mathcal{B}} 
\def\D{\mathcal{D}} 
\def\E{\mathcal{E}} 
\def\e{e} 
\def\G{\mathcal{G}} 
\def\H{\mathcal{H}} 
\def\I{\mathcal{I}} 
\def\L{\mathcal{L}} 
\def\N{\mathcal{N}} 
\def\S{\mathcal{S}} 

\def\ITr{ITr}
\def\ID{\mathcal{ID}}

\def\DomEnum{\textsc{Dom-Enum}}
\def\TransEnum{\textsc{Trans-Enum}}
\def\IDomEnum{\textsc{IDom-Enum}}
\def\ITransEnum{\textsc{ITrans-Enum}}
\def\DualEnum{\textsc{Dual-Enum}}
\def\Dual{\textsc{Dual}}

\DeclareMathOperator{\Min}{Min}
\DeclareMathOperator{\Max}{Max}
\DeclareMathOperator{\priv}{Priv}
\DeclareMathOperator{\poly}{poly}

\begin{document}

\title{On the dualization in distributive \\lattices and related problems\thanks{The first two authors have been supported by the ANR project GraphEn ANR-15-CE40-0009, France. The last author has been supported by JST CREST Grant Number JPMJCR1401, Japan.}}

\date{May 20, 2020}

\author{
Oscar Defrain\thanks{LIMOS, Université Clermont Auvergne, France.}~~
\addtocounter{footnote}{-1} 
\and
Lhouari Nourine\footnotemark
\and
~~~~Takeaki Uno\thanks{National Institute of Informatics, Japan.}
}

\maketitle

\begin{abstract}
In this paper, we study the dualization in distributive lattices, a generalization of the well-known hypergraph dualization problem.
We in particular propose equivalent formulations of the problem in terms of graphs, hypergraphs, and posets.
It is known that hypergraph dualization amounts to generate all minimal transversals of a hypergraph, or all minimal dominating sets of a graph.
In this new framework, a poset on vertices is given together with the input (hyper)graph, and minimal ``ideal solutions'' are to be generated.
This in particular allows us to study the complexity of the problem under various combined restrictions on graph classes and poset types, including bipartite, split, and co-bipartite graphs, and variants of neighborhood inclusion posets.
We for example show that while the enumeration of minimal dominating sets is possible with linear delay in split graphs, the problem, within the same class, gets as hard as for general graphs when generalized to this framework.
More surprisingly, this result holds even when the poset is only comparing vertices of included neighborhoods in the graph. 
If both the poset and the graph class are sufficiently restricted, we show that the dualization is tractable relying on existing algorithms from the literature.

\vskip5pt\noindent{}{\bf Keywords:} distributive lattice dualization, ideal enumeration, neighborhood inclusion posets, dominating sets, hypergraph transversals.
\end{abstract}

\def\figurescale{1.2}

\section{Introduction}

The dualization in Boolean lattices is a central problem in algorithmic enumeration as it is equivalent to the enumeration of the minimal transversals of a hypergraph, the minimal dominating sets of a graph, and to many other generation problems \cite{eiter1995identifying,kante2014enumeration}.
It~is also a problem of practical interest in database theory, logic, artificial intelligence and pattern mining \cite{kavvadias1993horn,eiter1995identifying,gunopulos1997data,eiter2003new,elbassioni2002algorithm,nourine2012extending}.
To~date, it is still open whether this problem can be solved in output-polynomial time. 
We say that an enumeration algorithm is running in {\em output-polynomial time} if its running time is bounded by a polynomial depending on the sizes of both the input and output data.
If the running times between two consecutive outputs is bounded by a polynomial depending on the size of the input, then we say that the algorithm is running with {\em polynomial delay}.
We refer the reader to \cite{johnson1988generating,creignou2019complexity,strozecki2019survey} for more details on the complexity of enumeration algorithms.
As of now, the best known algorithm for the dualization in Boolean lattices is due to Fredman and Khachiyan and runs in output quasi-polynomial time~\cite{fredman1996complexity}.
When generalized to arbitrary lattices, the problem is of practical interest in lattice-oriented machine learning through hypothesis generation~\cite{kuznetsov2004complexity,babin2017dualization}, and in pattern mining~\cite{nourine2012extending}.
It~was recently proved in \cite{babin2017dualization} that the dualization in this context is impossible in output-polynomial time, unless {\sf P=NP}.
In~\cite{defrain2019dualization}, it was shown that this result holds even when the premises in the implicational base (coding the lattice) are of size at most two.
In the case of premises of size one---when the lattice is distributive---the problem is still open.
The best known algorithm is due to Babin and Kuznetsov and runs in output sub-exponential time \cite{babin2017dualization}.
Output quasi-polynomial time algorithms are known for several subclasses, including distributive lattices coded by products of chains \cite{elbassioni2009algorithms}, or those coded by the ideals of an interval order \cite{defrain2019dualization}.

In this paper, we propose equivalent formulations for the dualization in distributive lattices, in terms of graphs, hypergraphs, and posets.
In the new framework, a poset on vertices is given together with the input (hyper)graph.
Then, the task is of enumerating minimal ideals of the poset with the desired property, i.e., transversality or domination.
We show that the obtained problems are equivalent to the dualization in distributive lattices, even when considering various combined restrictions on graph classes and poset types, including bipartite, split, and co-bipartite graphs, and variants of neighborhood inclusion posets; see Theorems~\ref{thm:maintrans} and~\ref{thm:maindom}.
Moreover, we believe that these equivalent problems may be simpler to attack using graph structure.
For combined restrictions on graph classes and poset types that are not considered in Theorems~\ref{thm:maintrans} and~\ref{thm:maindom}, we show that the problem gets tractable relying on existing algorithms from the literature; see Theorems~\ref{thm:split} and~\ref{thm:bipartite}.
A summary of these results is given in Figure~\ref{fig:sum}.

The rest of the paper is organized as follows.
In Section~\ref{sec:preliminaries} we introduce necessary concepts and definitions.
In Sections~\ref{sec:transideal} and~\ref{sec:domideal}, we generalize the two problems of enumerating minimal transversals and minimal dominating sets to the dualization in distributive lattices.
In Section~\ref{sec:tractable}, we exhibit tractable cases of the problem.
We discuss future work in Section~\ref{sec:conclusion}.

\section{Preliminaries}\label{sec:preliminaries}

We refer to~\cite{diestel2005graph} for graph terminology not defined below; all graphs considered in this paper are undirected, finite and simple. 
A {\em graph} $G$ is a pair $(V(G),E(G))$ where $V(G)$ is the set of {\em vertices} and $E(G)\subseteq \{\{x,y\} \mid x,y\in V(G), x\neq y\}$ is the set of {\em edges}.
Edges are usually denoted by $xy$ (or $yx$) instead of $\{x,y\}$.
A {\em clique} in a graph $G$ is a set of vertices $K$ such that every two vertices in $K$ are adjacent.
An {\em independent set} in a graph $G$ is a set of vertices $S$ such that no two vertices in $S$ are adjacent.
The {\em subgraph} of $G$ induced by $X\subseteq V(G)$, denoted by $G[X]$, is the graph $(X,E(G)\cap \{\{x,y\} \mid x,y\in X,\ x\neq y\})$; $G-X$ is the graph $G[V(G)\setminus X]$.
If $xy$ is an edge, $G-xy$ denotes the graph $(V(G),E(G)\setminus \{x,y\})$.
We note $N(x)$ the set of {\em neighbors} of $x$ defined by $N(x)=\{y\in V(G)\mid xy\in E(G)\}$.
We note $N[x]$ the set of {\em closed neighbors} of $x$ defined by $N[x]= N(x)\cup\{x\}$.
If it is not clear from the context, we add the subscript $G$ to denote the neighborhood in $G$, as in $N_G[x]$. 
Two vertices $x,y$ are called {\em twin} if $N[x]=N[y]$, and {\em false twin} if $N(x)=N(y)$.
For a given set $X\subseteq V(G)$, we respectively denote by $N[X]$ and $N(X)$ the two sets defined by $N[X]=\bigcup_{x\in X} N[x]$ and $N(X)=N[X]\setminus X$.

Let $G$ be a graph.
We say that $G$ is {\em bipartite} (resp.~{\em co-bipartite}) if $V(G)$ can be partitioned into two independent sets (resp.~two cliques).
If $V(G)$ can be partitioned into one independent set and one clique, then $G$ is called {\em split}.

Let $D,X\subseteq V(G)$ be two subsets of vertices of $G$.
We say that $D$ \emph{dominates} $X$ if $X\subseteq N[D]$.
It is (inclusion-wise) minimal if $X\not\subseteq N[D\setminus \{x\}]$ for any $x\in D$.
A (minimal) \emph{dominating set} of $G$ is a (minimal) dominating set of $V(G)$.
The set of all minimal dominating sets of $G$ is denoted by $\D(G)$, and the problem of enumerating $\D(G)$ given $G$ is denoted by \DomEnum{}.
The set of all minimal dominating sets of a given subset $X$ of vertices of $G$ is denoted by $\D_G(X)$.
Let $x$ be a vertex of $D$.
We say that $x$ has {\em private neighbor} $y$ w.r.t.~$D$ in $G$ if $y\in N[D]$ and $y\not\in N[D\setminus \{x\}]$. 
The set of private neighbors of $x\in D$ in $G$ is denoted by $\priv(D,x)$. 
It is easy to see that $D$ is a minimal dominating set of $G$ if and only if $D$ dominates $G$ and $\priv(D,x)\neq \emptyset$ for every $x\in D$.
Also, note that a set in $\D_G(X)$ may contain vertices in $G-X$ as long as these vertices have private neighbors in~$X$.

We refer to~\cite{berge1984hypergraphs} for hypergraph terminology not defined below.
A {\em hypergraph} $\H$ is a pair $(V(\H),\E(\H))$ where $V(\H)$ is the set of {\em vertices} (or {\em groundset}) and $\E(\H)$ is a set of non-empty subsets of $V(\H)$ called {\em edges} (or {\em hyperedges}).
In this paper, and as is custom, we will often describe a hypergraph by its set of edges only, and will denote $\e\in \H$ in place of $\e\in \E(\H)$.
If $x$ is a vertex of $\H$, we denote by $\E_x$ the set of edges incident to $x$ defined by $\E_x=\{\e\in \E(\H) \mid x\in \e\}$.
A {\em transversal} in a hypergraph $\H$ is a set of vertices $T$ that intersects every edge of $\H$.
It is minimal if it does not contain any transversal as a proper subset.
The set of all minimal transversals of $\H$ is denoted by $Tr(\H)$, and the problem of enumerating $Tr(\H)$ given $\H$ is denoted by \TransEnum{}.
A hypergraph $\H$ is called {\em Sperner} if $\e\not\subseteq \e'$ for any two distinct hyperedges $\e,\e'\in \H$.
It is well known that hypergraphs can be considered Sperner when dealing with \TransEnum{}.
In the following, we denote by $\N(G)$ the Sperner hypergraph of closed neighborhoods of $G$ defined by $\N(G)=\Min_\subseteq\{N[x]\mid x\in V(G)\}$.
It is not hard to see that \DomEnum{} is a particular case of \TransEnum{}, where the minimal dominating sets of $G$ are exactly the minimal transversals of $\N(G)$.
Recently in \cite{kante2014enumeration}, it was shown that the two problems are in fact equivalent, even when restricted to co-bipartite graphs.
The result in \cite{kante2014enumeration} is based on the following construction.
For~any hypergraph $\H$, the {\em bipartite incidence graph} of $\H$ is the graph $I(\H)$ with bipartition $X=V(\H)$ and $Y=\{y_\e \mid \e\in \E(\H)\}$, and where there is an edge between $x\in X$ and $y_\e\in Y$ if $x$ belongs to~$\e$ in $\H$.
The construction of a bipartite incidence graph is given in Figure~\ref{fig:inc-bipartite}.

\begin{figure}[t]
  \center
  \includegraphics[scale=\figurescale]{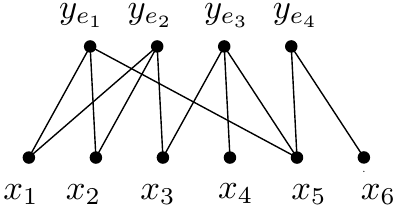}
  \caption{The bipartite incidence graph $I(\H)$ of bipartition $X=V(\H)$ and $Y=\{y_\e \mid \e\in \H\}$ for the hypergraph $\H=\{\e_1,\e_2,\e_3,\e_4\}$ where $\e_1=\{x_1,x_2,x_5\}$, $\e_2=\{x_1,x_2,x_3\}$, $\e_3=\{x_3,x_4,x_5\}$ and $\e_4=\{x_5,x_6\}$. Then $xy_e\in E(I(\H))$ if and only if $x\in e$.}\label{fig:inc-bipartite}
\end{figure}

A {\em partial order} on a set $X$ (or {\em poset}) is a binary relation $\leq$ on $X$ which is reflexive, anti-symmetric and transitive, denoted by $P=(X,\leq)$.
Two elements $x,y$ of $X$ are said to be {\em comparable} if $x \leq y$ or $y \leq x$, otherwise they are said to be {\em incomparable}.
If $x<y$ and there is no element $z$ such that $x<z<y$ then we say that $y$ {\em covers} $x$.
Posets are represented by their \emph{Hasse diagram} in which each element is a vertex in the plane, and where there is a line segment or curve that goes upward from $x$ to $y$ whenever $y$ covers $x$.
See Figure~\ref{fig:dual} for an example.
A subset of a poset in which every pair of elements is comparable is called a {\em chain}. 
A subset of a poset in which no two distinct elements are comparable is called an {\em antichain}.
A poset is an {\em antichain poset} (resp.~\emph{total order}) if the set of its elements is an antichain (resp.~chain).
We call poset induced by $S\subseteq X$, denoted $P[S]$, the suborder restricted on the elements of $S$ only; $P-S$ is the poset $P[X\setminus S]$.
A~set $I\subseteq X$ is an {\em ideal} of $P$ if $x\in I$ and $y\leq x$ imply $y\in I$. 
If $x\in I$ and $x\leq y$ imply $y\in I$, then $I$ is called {\em filter} of $P$.
Note that the complement of an ideal is a filter, and vice versa. 
For~every $x\in P$ we associate the {\em principal ideal of $x$} (or simply {\em ideal of $x$}) denoted by $\downarrow\,x$ and defined by $\downarrow x=\{y\in X \mid y\leq x\}$. 
The {\em principal filter of} $x\in X$ is the dual $\uparrow x=\{y\in X \mid x\leq y\}$.
The set of all subsets of $X$ is denoted by $2^X$, and the set of all ideals of $P$ by~$\mathcal{I}(P)$.
Clearly, $\I(P)\subseteq 2^X$ and $\I(P)=2^X$ whenever $P$ is an antichain poset.
If $S$ is a subset of $X$, we respectively denote by $\downarrow S$ and $\uparrow S$ the two sets defined by $\downarrow S=\bigcup_{x\in S} \downarrow x$ and $\uparrow S=\bigcup_{x\in S} \uparrow x$. 
We~respectively denote by $\Min(S)$ and $\Max(S)$ the sets of minimal and maximal elements of $S$ w.r.t.~$\leq$.

The next definition is central in this paper.

\begin{definition}
Let $P=(X,\leq)$ be a partial order and $B^+$, $B^-$ be two antichains of $P$.
We~say that $B^+$ and $B^-$ are {\em dual} in $P$ whenever $\downarrow B^+\,\cup \uparrow B^-=X$ and $\downarrow B^+\,\cap \uparrow B^-=\emptyset$.
\end{definition}

Note that deciding whether two antichains $B^+$ and $B^-$ of $P$ are dual can be done in polynomial time in the size of $P$ by checking whether ${B^-=\Min(P-\!\downarrow B^+)}$, or equivalently if ${B^+=\Max(P-\!\uparrow B^-)}$.
Notations $B^+$ and $B^-$ in fact come from these equalities.
However, the task becomes difficult when the poset is not fully given, but only an implicit coding---of possibly logarithmic size in the size of $P$---is given.
This is usually the case when considering dualization problems in lattices.

A {\em lattice} is a poset in which every two elements have a {\em supremum} (also called {\em join}) and a {\em infimum} (also called a {\em meet}); see \cite{davey2002introduction,gratzer2011lattice}.
In this paper however, only the next two characterizations from \cite{birkhoff1940lattice} will suffice.
We denote by {\em Boolean lattice} any poset isomorphic to $(2^X,\subseteq)$ for some set $X$; such a lattice is also called {\em hypercube}.
We denote by {\em distributive lattice} any poset isomorphic to $(\I(P),\subseteq)$ for some partially ordered set $P=(X,\leq)$.
Then, $X$ and $P$ are called {\em implicit coding} of the lattice and we denote by $\L(X)$ and $\L(P)$ the two lattices coded by $X$ and $P$.
Clearly, every Boolean lattice is a distributive lattice where $P$ is an antichain poset, as $\I(P)=2^X$ for such $P$. 
In~fact, it can be easily seen that each comparability $x\leq y$ in $P$ removes from $(2^X,\subseteq)$ the Boolean lattice given by the interval $[y,X\setminus \{x\}]$, i.e., the elements containing $y$ but not $x$.
At last, observe that $\L(P)$ may be of exponential size in the size of $P$: this is in particular the case when the lattice is Boolean, i.e., when $P$ is an antichain poset.
An example of a distributive lattice coded by the ideals of a poset is given in Figure~\ref{fig:dual}.

\begin{figure}
  \center
  \includegraphics[scale=\figurescale]{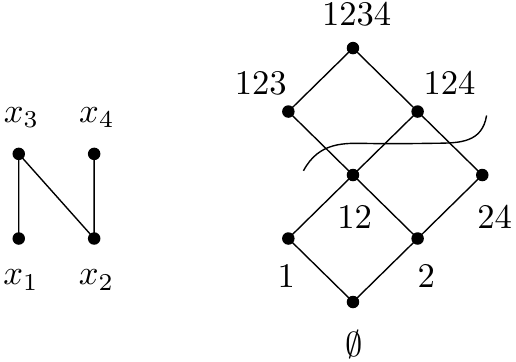}
  \caption{A poset $P=(X,\leq)$ (left) that codes the lattice $\L(P)=(\I(P),\subseteq)$ (right), and the border (curved line) formed by the two dual antichains $\B^+=\{\{x_1,x_2\}, \{x_2,x_4\}\}$ and $\B^-=\{\{x_1,x_2,x_3\}, \{x_1,x_2,x_4\}\}$ of $\L(P)$.
  For better readability, ideals are denoted by the indexes of their elements in the lattice, i.e., $123$ stands for $\{x_1,x_2,x_3\}$.}
  \label{fig:dual}
\end{figure}

In this paper, we are concerned with the following decision problem and one of its two generation versions.

\begin{decproblem}
  \problemtitle{Dualization in Distributive Lattices given by the Ideals of a Poset (\Dual{})}
  \probleminput{A poset $P=(X,\leq)$ and two antichains $\B^+,\B^-$ of $\L(P)$.}
  \problemquestion{Are $\B^+$ and $\B^-$ dual in $\L(P)$?}
\end{decproblem}

\begin{problemgen}
  \problemtitle{Generation version of \Dual{} (\DualEnum{})}
  \probleminput{A poset $P=(X,\leq)$ and an antichain $\B^+$ of $\L(P)$.}
  \problemquestion{The dual antichain $\B^-$ of $\B^+$ in $\L(P)$.}
\end{problemgen}

We stress the fact that the lattice $\L(P)$ is not given in any of the two problems defined above.
Only $P$ is given, which is a crucial point.
Hence, \DualEnum{} can be reformulated without any mention of the lattice, namely as the enumeration of all inclusion-wise minimal ideals of $P$ that are not a subset of any ideal in $\B^+$, i.e., as the enumeration of the set 
\[
  \B^-=\Min_\subseteq\{I\in \I(P) \mid I\not\subseteq B~\text{for any}~B\in \B^+\}.
\]
Then, computing a first solution to this problem is easy, as we start from $I=X$ as an ideal, and remove its maximal elements until it is a minimal ideal such that $I\not\subseteq B$ for any $B\in\B^+$.
However, it is still open whether the problem can be solved in output quasi-polynomial time.
To~date, the best known algorithm runs in output sub-exponential time $2^{O(n^{0,67} \log^3 N)}$ where $N=|\B^+|+|\B^-|$, and where $P$ is given as a $n\times n$ matrix \cite{babin2017dualization}.
Output quasi-polynomial time algorithms running in time $\poly(N,n) + N^{o(\log N)}$ are known for several subclasses, including distributive lattices coded by products of chains \cite{elbassioni2009algorithms}, or distributive lattices coded by the ideals of an interval order \cite{defrain2019dualization}.

If the poset is an antichain, i.e., if the lattice is Boolean, then this problem calls for enumerating every inclusion-wise minimal subset of $X$ that is not a subset of any $B\in\B^+$, or equivalently, that intersects every set in $\H=\{X\setminus B \mid B \in \B^+\}$. 
Under such a formulation, it is easily seen that the dualization in Boolean lattices is equivalent to \TransEnum{} (hence to \DomEnum{}), where $Tr(\H)=\B^-$; see~\cite{nourine2014dualization,nourine2016encyclopedia}.
In~this case, the best known algorithm runs in output quasi-polynomial time $N^{o(\log N)}$ where $N=|\B^+|+|\B^-|$, and the existence of an output-polynomial time algorithm remains open after decades of research \cite{eiter1995identifying,fredman1996complexity,eiter2008computational}.
Due to its equivalence with \DomEnum{}, the complexity of the dualization in Boolean lattice has been precised under various restrictions on graph classes and parameters.
Among these results, output-polynomial time algorithms were given for degenerate \cite{eiter2003new}, line \cite{kante2012neighbourhood,golovach2015incremental}, split \cite{kante2014enumeration}, chordal \cite{kante2015chordal}, triangle-free graphs \cite{bonamy2019triangle}, graphs of bounded clique-width \cite{courcelle2009linear}, LMIM-width \cite{Golovach2018}, etc.
Other classes of graphs remain open, including co-bipartite (as the problem in this case is as hard as \TransEnum{}, hence as hard as in general graphs \cite{kante2014enumeration}), or unit disk graphs~\cite{kante2008minimal,golovach2016enumerating}.

The aim of this work is to generalize \TransEnum{} and \DomEnum{} to the dualization in distributive lattices, in order to obtain similar finer characterizations on the difficulty of the later problem, using (hyper)graph parameters.

\section{Transversal-ideals}\label{sec:transideal}

We generalize \TransEnum{} to the enumeration of the minimal ideals of a poset with the transversal property.
We show that the obtained problem is equivalent to the dualization in distributive lattices.

Let $\H$ be a hypergraph and $P_\H$ be a partial order on vertices of $\H$.
Let $I$ be a subset of vertices of $\H$.
We say $I$ is a {\em transversal-ideal} of $\H$ w.r.t.~$P_\H$ if it is an ideal of $P_\H$, and a transversal of~$\H$.
It~is minimal if it does not contain any transversal-ideal as a proper subset.
We denote by $\ITr(\H,P_\H)$ the set of minimal transversal-ideals of $\H$ w.r.t.~$P_\H$, and define the problem of generating $\ITr(\H,P_\H)$ as follows.

\begin{problemgen}
  \problemtitle{Minimal Transversal-Ideals Enumeration (\ITransEnum{})}
  \probleminput{A hypergraph $\H$ and a partial order $P_\H$ on vertices of $\H$.}
  \problemquestion{The set $\ITr(\H,P_\H)=\Min_\subseteq\{I\in \I(P_\H) \mid I\ \text{is a transversal of}\ \H\}$.}
\end{problemgen}

\begin{figure}
  \center
  \includegraphics[page=1,scale=\figurescale]{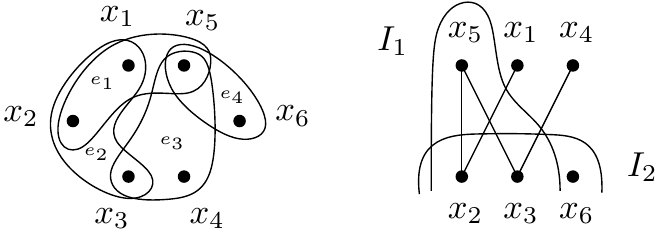}
  \caption{A hypergraph $\H=\{e_1,e_2,e_3,e_4\}$ (left) and a partial order $P_\H$ on vertices of $\H$ (right), where $e_1=\{x_1,x_2,x_5\}$, $e_2=\{x_1,x_2,x_3\}$, $e_3=\{x_3,x_4,x_5\}$ and $e_4=\{x_5,x_6\}$.
  The minimal transversal-ideals for this instance are $I_1=\{x_2,x_3,x_5\}$ and $I_2=\{x_2,x_3,x_6\}$.
  Note that $I_1$ is the ideal of two minimal transversals $T_1=\{x_2,x_5\}$ and $T_2=\{x_3,x_5\}$.}
  \label{fig:running-example-H}
\end{figure}

An instance of this problem is given in Figure~\ref{fig:running-example-H}.
Observe that as for \DualEnum{}, computing a first solution to \ITransEnum{} is easily done by starting with $I=V(\H)$ as a transversal-ideal, and greedily reducing it until it is minimal.
It is worth pointing out that in the case where $P_\H$ is an antichain poset, then the minimal transversal-ideals of $\H$ w.r.t.~$P_\H$ are exactly the minimal transversals of $\H$, and the two problems \ITransEnum{} and \TransEnum{} are equivalent.
In the general case, however, a minimal transversal-ideal of $\H$ w.r.t.~$P_\H$ may contain several minimal transversals of $\H$; see Figure~\ref{fig:running-example-H} for an example.
If~$P_\H$ is a total order, then $\H$ admits a unique minimal transversal-ideal no matter the number of minimal transversals.
Consequently, the size of $Tr(\H)$ may be exponential in the size of $Tr(\H,P_\H)$.

It easily observed that \DualEnum{} is a particular case of \ITransEnum{}, where ${P=P_\H=(X,\leq)}$ and where to every ideal $B\in \B^+$ corresponds a hyperedge $\e=X\setminus B$ of $\H$, as in that case ideal $I$ is a transversal-ideal of $\H$ (i.e., $I\cap\e\neq\emptyset$ for all $\e\in \H$) if and only if $I\not\subseteq B$ for any $B\in \B^+$.
In other words, \DualEnum{} appears as a particular case of \ITransEnum{} in which $\H$ defines a collections of filters of $P_\H$.
However, $\H$ may not be a collection of filters of $P_\H$ in general, and \ITransEnum{} may appear as a tougher problem at first glance.
Nevertheless, we~show that the two problems are equivalent by showing that hypergraphs that do not share this property can be closed in the poset with no impact on the solutions to enumerate.

For any hypergraph $\H$ and poset $P_\H$, we denote by $\uparrow \H$ the {\em filter-closed hypergraph} of $\H$ w.r.t.~$P_\H$ defined by $V(\uparrow\H)=V(\H)$ and $\E(\uparrow\H)=\Min_\subseteq\{\uparrow \e\mid \e\in \E(\H)\}$.
Observe that $|\uparrow \H|\leq |\H|$.

\begin{lemma}\label{lemma:hypergraph-closure}
  Let $I$ be an ideal of $P_\H$.
  Then $I$ is a transversal-ideal of $\H$ if and only if it is a transversal-ideal of $\uparrow \H$.
  In particular, $\ITr(\H,P_\H)=\ITr(\uparrow \H, P_\H)$.
\end{lemma} 

\begin{proof}
  Let $I$ be an ideal of $P_\H$ and $\e$ be an edge of $\H$.
  We show that $I$ intersects $\e$ if and only if it intersects $\uparrow \e$.
  Clearly if $I$ intersects $\e$ then it intersects $\uparrow \e$ as $\e\subseteq \uparrow \e$.
  Let us assume that $I$ intersects $\uparrow \e$ and let $x\in I\cap \uparrow \e$.
  Then there exists $y\in \e$ such that $y\leq x$.
  Since $I$ is an ideal, $y\in I$.
  Thus $I\cap \e\neq\emptyset$.
  Hence $\ITr(\H,P_\H)=\ITr(\uparrow \H, P_\H)$.
\end{proof}

\begin{lemma}\label{lemma:incident-edge-incl}
  If $\uparrow \H=\H$ then $x\leq y$ implies $\E_x\subseteq \E_y$ for all $x,y\in V(\H)$.
\end{lemma} 

\begin{proof}
  Let $\H$ such that $\uparrow \H=\H$, $x,y\in V(\H)$ such that $x\leq y$, and $E\in \E_x$.
  Since $E=\uparrow E$ and $x\leq y$, $y\in E$.
  Hence the desired result.
\end{proof}

In what follows, we say that $P_\H$ is a poset of {\em incident edge inclusion} of $\H$ if $x\leq y$ implies $\E_x\subseteq \E_y$.
By Lemma~\ref{lemma:incident-edge-incl}, every partial order $P_\H$ is a poset of incident edge inclusion of $\uparrow\H$.
We conclude with the following result.

\begin{theorem}\label{thm:maintrans}
  \DualEnum{} and \ITransEnum{} are equivalent, even when restricted to posets of incident edge inclusion.
\end{theorem}

\begin{proof}
  It follows from the equivalence $I\not\subseteq B~\text{for any}~B\in \B^+$ if and only if $I\cap {(X\setminus B)\neq\emptyset}$ $\text{for all}~B\in \B^+$, that \DualEnum{} is a particular case of \ITransEnum{}, where $\H=\{X\setminus B \mid B\in \B^+\}$ and $\ITr(\H,P_\H)=\B^-$.

  We show that \ITransEnum{} reduces to \DualEnum{}.
  Let $(\H,P_\H)$ be an instance of the first problem and $\G=\uparrow \H$ be the filter-closed hypergraph of $\H$ w.r.t.~$P_\H$.
  Clearly, $\G$ can be computed in polynomial time in the sizes of $\H$ and $P_\H$, and $\B^+=\{X\setminus \e \mid \e\in \E(\G)\}$ defines an antichain of $\L(P_\H)$.
  By Lemma~\ref{lemma:hypergraph-closure}, $\ITr(\H,P_\H)=\ITr(\G,P_\H)$.
  As~$\ITr(\G,P_\H)=\{I\in \I(P_\H) \mid I\not\subseteq B~\text{for any}~B\in \B^+\}$, we deduce that $\B^-=\ITr(\G,P_\H)$ where $\B^-$ is the dual antichain of $\B^+$ in $\L(P_\H)$.
  Hence that \ITransEnum{} can be solved using an algorithm for \DualEnum{} on $P_\H$ and $\B^+$.
\end{proof}

\section{Dominating-ideals}\label{sec:domideal}

We generalize \DomEnum{} to the enumeration of the minimal ideals of a poset with the domination property.
We show that the obtained problem is equivalent to the dualization in distributive lattices.
This will allow us to study the complexity of the problem under various restrictions on graph classes and poset types.

Let $G$ be a graph and $P_G$ be a partial order on vertices of $G$. 
Let $D$ be a subset of vertices of $G$.
We say that $D$ is a {\em dominating-ideal} of $G$ w.r.t.~$P_G$ if it is an ideal of $P_G$ and a dominating set of~$G$.
It~is minimal if it does not contain any dominating-ideal as a proper subset.
Note that a dominating-ideal $I$ is minimal if and only if $\priv(I,x)\neq \emptyset$ for all $x\in \Max(I)$. 
We denote by $\ID(G,P_G)$ the set of minimal dominating-ideals of $G$ w.r.t.~$P_G$, and define the problem of generating $\ID(G,P_G)$ as follows.

\begin{problemgen}
  \problemtitle{Minimal dominating-ideals enumeration (\IDomEnum{})}
  \probleminput{A graph $G$ and a partial order $P_G$ on vertices of $G$.}
  \problemquestion{The set $\ID(G,P_G)=\Min_\subseteq\{I\in \I(P_G) \mid I\ \text{dominates}\ G\}$.}
\end{problemgen}

\begin{figure}
  \center
  \includegraphics[page=2,scale=\figurescale]{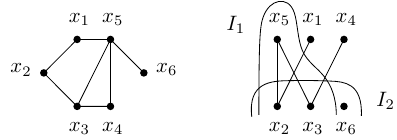}
  \caption{A graph $G$ (left) and a partial order $P_G$ (right) on vertices of $G$ such that $\N(G)=\H$, where $\H$ is the hypergraph defined in Figure~\ref{fig:running-example-H}.
  The minimal dominating-ideals for this instance are $I_1=\{x_2,x_3,x_5\}$ and $I_2=\{x_2,x_3,x_6\}$.}
  \label{fig:running-example-G}
\end{figure}

An instance of this problem is given in Figure~\ref{fig:running-example-G}.
Observe that as for the classical case when $P$ is an antichain, \IDomEnum{} naturally appears as a particular case of \ITransEnum{} where $\ID(G,P_G)=\ITr(\N(G),P_G)$.
The rest of this section is devoted to the proof of their equivalence.

In the following, we say that $P_G$~is a {\em neighborhood inclusion poset} of~$G$ if $x\leq y$ implies $N[x]\subseteq N[y]$, and that $P_G$~is a {\em weak neighborhood inclusion poset} of~$G$ if at least one of $N[x]\subseteq N[y]$ and $N[x]\supseteq N[y]$ holds whenever $x\leq y$.
Clearly, every neighborhood inclusion poset is a weak neighborhood inclusion poset.
It can be seen that the first restriction is closely related to the one of Lemma~\ref{lemma:incident-edge-incl}, as to every neighborhood inclusion poset of a graph corresponds an incident edge inclusion poset in $\N(G)$.
Henceforth, neighborhood inclusion posets naturally appear when considering dualization problems in distributive lattices.

The aforementioned equivalences are the following.

\begin{theorem}\label{thm:maindom}
  \ITransEnum{} and \IDomEnum{} are equivalent, even when restricted to:
  \begin{enumerate}
    \item bipartite graphs;\label{item:main1}
    \item split graphs and weak neighborhood inclusion posets; and\label{item:main2}
    \item co-bipartite graphs and neighborhood inclusion posets.\label{item:main3}
  \end{enumerate}
\end{theorem}

\begin{proof}
  Clearly, \IDomEnum{} is a particular case of \ITransEnum{} where $\ID(G,P_G)=\ITr(\N(G),P_G)$.

  We show that \ITransEnum{} reduces to \IDomEnum{}.
  Let $(\H,P_\H)$ be a non-trivial instance (such that $\H\neq\emptyset$) of \ITransEnum{}.
  Note that by Lemma~\ref{lemma:hypergraph-closure}, we can restrict ourselves to the case where $\H=\uparrow \H$. 
  Hence by Lemma~\ref{lemma:incident-edge-incl}, $x\leq y$ in $P_\H$ implies $\E_x\subseteq \E_y$.

  Consider the bipartite incidence graph $I(\H)$ of $\H$ of bipartition $X=V(\H)$ and $Y=\{y_\e \mid \e\in \E(\H)\}$, where $xy_\e\in E(I(\H))$ if and only if $x\in X$, $y_\e \in Y$ and $x\in \e$; see~Section~\ref{sec:preliminaries} and Figure~\ref{fig:inc-bipartite}.
  A first observation is the following:

  \begin{observation}\label{obs:openneighborhoodinclusion}
  Let $x,y\in P_\H$.
  Then $x\leq y$ implies $N(x)\subseteq N(y)$ in $I(\H)$.
  \end{observation}

  The remainder of the proof is separated into three parts: we will adapt the construction of the bipartite incidence graph according to each item of the theorem.
  
  Let us first consider Item~\ref{item:main1}.
  Let $G$ be the graph obtained from $I(\H)$ by adding a single vertex $v$ connected to every vertex of $X$.
  Then $G$ is bipartite with bipartition $X$ and $Y\cup\{v\}$.
  Let $P_G$ be the poset obtained from $P_\H$ by making every $y\in Y$ greater than every $x\in X$, and $v$ incomparable with every other vertex, i.e., $P_G= P_\H \cup \{x<y \mid x\in X,\ y\in Y\}$.
  We prove the following.

  \begin{claim}\label{claim:bipartite}
    Let $I\subsetneq V(\H)$.
    Then $I\in \ITr(\H,P_\H)$ if and only if $I\cup\{v\}\in \ID(G,P_G)$.
  \end{claim}

  \begin{proof}[Proof of the claim]
    Let $I\subsetneq V(\H)$ such that $I\in \ITr(\H,P_\H)$.
    As $\H$ is non-empty, $I\neq \emptyset$.
    By construction, $I$ is an ideal of $P_G$ and it is a minimal dominating-ideal of subset $Y$, i.e., $Y\subseteq N[I]$ and $Y\not\subseteq N[I\setminus \{x\}]$ for any $x\in \Max(I)$.
    By hypothesis $I\neq V(\H)$, hence $I$ does not dominate $G$, and $I\cup\{v\}$ does; $v$ is here to dominate elements of $X$ that are not in the transversal.
    Since $v$ is not adjacent to $Y$, it does not steal private neighbors to vertices in $I$. Hence $I\cup \{v\}$ is a minimal dominating-ideal of $G$.
    Let $I\subsetneq V(\H)$ such that $I\in \ID(G,P_G)$.
    Note that $y\not\in I$ for any $y\in Y$ as $X\subseteq \downarrow y$ and $X$ dominates $G$. 
    As $\H$ is non-empty, $I\cap X\neq\emptyset$.
    By hypothesis, $I\neq X$.
    Thus $v$ has a private neighbor in $X$ and $\priv(I,x)\subseteq Y$ for all $x\in I$.
    Hence $I$ is a minimal transversal-ideal of $\H$.
    \cqed
  \end{proof}

  Let us now consider Item~\ref{item:main2}.
  Let $G$ be the graph obtained from $I(\H)$ by completing $X$ into a clique, and by adding a single vertex $v$ connected to every vertex of the graph, i.e., $v$ is universal in $G$.
  Then $G$ is split with clique $X\cup\{v\}$ and independent set~$Y$.
  Let $P_G$ be the poset obtained from $P_\H$ by making every $y\in Y$ greater than~$v$, i.e., $P_G=P_\H \cup \{v<y \mid y\in Y\}$.
  We prove the following two claims.

  \begin{claim}\label{claim:weak}
    $P_G$ is a weak neighborhood inclusion poset on~$G$.
  \end{claim}

  \begin{proof}[Proof of the claim]
    Clearly, $x\leq y$ either implies $x,y\in X$, or both $x=v$ and $y\in Y$. 
    In~the first case, it follows from Observation~\ref{obs:openneighborhoodinclusion} that $N[x]\subseteq N[y]$ as $X\cup v$ induces a clique.
    In~the other case, $N[v]\supseteq N[y]$ as $v$ is universal.
    \cqed
  \end{proof}

  \begin{claim}\label{claim:split}
    Let $I\subseteq V(\H)$.
    Then $I\in \ITr(\H,P_\H)$ if and only if $I\in \ID(G,P_G)$, $I\neq \{v\}$.
  \end{claim}

  \begin{proof}[Proof of the claim]
    Let $I\subseteq V(\H)$ such that $I\in \ITr(\H,P_\H)$.
    As $\H$ is non-empty, $I\neq \emptyset$.
    By construction, $I$ is an ideal of $P_G$, $I\neq\{v\}$, and it is a minimal dominating-ideal of $Y$.
    As $I$ dominates $X\cup\{v\}$, it is a minimal dominating-ideal of $G$.
    Let $I\subseteq V(\H)$ such that $I\in \ID(G,P_G)$ and $I\neq \{v\}$.
    Note that $y\not\in I$ for any $y\in Y$ as $v\in \downarrow y$ and $v$ dominates $G$.
    Since $I\neq \{v\}$, $I\subseteq X$. 
    Since $X$ induces a clique, $\priv(I,x)\subseteq Y$ for all $x\in I$.
    Hence $I$ is a minimal transversal-ideal of $\H$.
    \cqed
  \end{proof}

  We now consider Item~\ref{item:main3}.
  Let $G$ be the graph obtained from $I(\H)$ by adding a single vertex $v$ connected to every vertex of $X$, and by completing both $X$ and $Y$ into cliques.
  Then $G$ is co-bipartite with cliques $X\cup\{v\}$ and $Y$.
  Let $P_G=P_\H$.
  We prove the following two claims.

  \begin{claim}\label{claim:strong}
    $P_G$ is a neighborhood inclusion poset on~$G$.
  \end{claim}

  \begin{proof}[Proof of the claim]
    Let $x,y\in P_G$ such that $x\leq y$.
    It follows from Observation~\ref{obs:openneighborhoodinclusion} that $N[x]\subseteq N[y]$ as $X\cup v$ induces a clique.
    \cqed
  \end{proof}

  \begin{claim}\label{claim:co-bipartite}
    Let $I\subseteq V(\H)$.
    Then $I\in \ITr(\H,P_\H)$ if and only if $I \in\ID(G,P_G)$ and $I\not\in \{\{x,y\} \mid x\in X\cup \{v\},\ y\in Y\}$.
  \end{claim}

  \begin{proof}[Proof of the claim]
    Let $I\subseteq V(\H)$ such that $I\in \ITr(\H,P_\H)$.
    As $\H$ is non-empty, $I\neq \emptyset$.
    By construction, $I$ is an ideal of $P_G$, $I\not\in \{\{x,y\} \mid x\in X\cup \{v\}$, $y\in Y\}$, and it is a minimal dominating-ideal of $Y$.
    As $I$ dominates $X\cup\{v\}$, it is a minimal dominating-ideal of $G$.
    Let $I\subseteq V(\H)$ such that $I\in \ID(G,P_G)$ and $I\not\in \{\{x,y\} \mid x\in X\cup \{v\}$, $y\in Y\}\}$.
    Note that $y\not\in I$ for any $y\in Y$ or else, as $v$ is non adjacent to any vertex in $Y$, $I$ must contain one vertex of $X\cup \{v\}$ to dominate $G$.
    Then $I$ does not contain any other vertex as it dominates $G$, and $I\in \{\{x,y\} \mid x\in X\cup \{v\}$, $y\in Y\}\}$, a case excluded by hypothesis.
    Moreover, $v\not\in I$ as otherwise, $I$ must contain some $x\in X$ to dominate $Y$ and $N[v]\subseteq N[x]$.
    Hence $I\subseteq X$.
    Since $X$ induces a clique, $\priv(I,x)\subseteq Y$ for all $x\in I$.
    Hence $I$ is a minimal transversal-ideal of $\H$.
    \cqed
  \end{proof}

  The proof of the theorem follows from Claims~\ref{claim:bipartite}, \ref{claim:weak}, \ref{claim:split}, \ref{claim:strong} and \ref{claim:co-bipartite}, observing that $G=I(\H)$ is constructed in polynomial time in the sizes of $\H$ and $P_\H$, and that $\ITr(\H,P_\H)$ can be enumerated with polynomial delay from $\ID(G,P_G)$ on the constructed graph and poset.
  Indeed, in the case of Item~\ref{item:main1} only one extra solution (namely $I=V(\H)$) has to be handled separately.
  In the case of Item~\ref{item:main2}, only one solution (namely $I=\{v\}$) has to be discarded.
  In the case of Item~\ref{item:main3}, at most $|V(G)|^2$ solutions (namely every subsets of $V(G)$ of size two) have to be discarded.
  This concludes the proof.
\end{proof}

\section{Tractable cases for dominating-ideals enumeration}\label{sec:tractable}

In the following, we show that combined restrictions left by Theorem~\ref{thm:maindom} are tractable (see Figure~\ref{fig:sum}), using existing algorithms and techniques from the literature for the enumeration of minimal dominating sets in split and triangle-free graphs \cite{kante2014enumeration,bonamy2019triangle}.
Our results rely on the following important property.

\begin{proposition}\label{prop:domantichain}
  Let $G$ be a graph and $P_G$ be a weak neighborhood inclusion poset on~$G$.
  Then, every minimal dominating set of $G$ is an antichain of $P_G$.
  Hence there is a bijection between minimal dominating sets of $G$ and their ideal in $P_G$.
  If in addition $P_G$ is a neighborhood inclusion poset, then $\Max(D)$ dominates $G$ whenever $D$ does.
\end{proposition}

\begin{proof}
  Let $D$ be a dominating set of $G$ and $x,y\in D$ be two comparable elements of $P_G$.
  If~$P_G$ is a weak neighborhood inclusion poset, then either $N[x]\subseteq N[y]$ or $N[x]\supseteq N[y]$.
  Thus, either $D\setminus \{x\}$ or $D\setminus \{y\}$ dominates $G$ and we deduce that every minimal dominating set of $G$ is an antichain of $P_G$.
  If $P_G$ is a neighborhood inclusion poset and $x\leq y$, then $D\setminus \{x\}$ dominates $G$ and we deduce that $\Max(D)$ dominates $G$.
  Since the set of antichains and the set of ideals of a poset are in bijection, we conclude to a bijection between minimal dominating sets of $G$ and their ideal in $P_G$.
\end{proof}

A consequence of Proposition~\ref{prop:domantichain} is the following equality.
\begin{equation}\label{eqn:main}
    \ID(G,P_G)=\Min_\subseteq\{\downarrow D \mid D\in \D(G)\}.\tag{$5$}
\end{equation}

Note that instances that verify this property are not trivially tractable, as two of the constructed instances in the proof of Theorem~\ref{thm:maindom} satisfy Proposition~\ref{prop:domantichain}, despite the fact that the problem on such instances is \DualEnum{}-hard, hence \TransEnum{}-hard.

\subsection{Split graphs and neighborhood inclusion posets}

In \cite{kante2014enumeration}, the authors give a polynomial delay algorithm to enumerate minimal dominating sets in split graphs.
Their algorithm relies on the two observations that if $G$ is a split graph of maximal independent set $S$, and clique $C$, then the set of intersections of minimal dominating sets of $G$ with $C$ is in bijection with $\D(G)$, and it forms an independence system.
A pair $(X,\S)$ where $\S\subseteq 2^X$ is an {\em independence system} if $\emptyset\in \S$ and if $S\in \S$ implies that $S'\in \S$ for all $S'\subseteq S$.
We show that these observations can be generalized in our case, giving a polynomial delay algorithm to enumerate $\ID(G,P_G)$ whenever $G$ is split and $P_G$ is a neighborhood inclusion poset.

In what follows, we follow the notations of \cite{kante2014enumeration} to denote the intersection of a dominating set $D$ with some set $W\subseteq V(G)$, namely $D_W=D\cap W$.
We extend this notation to the set of minimal dominating sets as follows: 
\[
  \D_W(G)\eqdef\{D_W \mid D\in \D(G)\}.
\]

\begin{proposition}[\cite{kante2014enumeration}]\label{prop:kante-characterization}
  Let $G$ be a split graph with maximal independent set $S$, clique~$C$, and let $D$ be a minimal dominating set of $G$.
  Then $D_S=S\setminus N(D_C)$.
  Furthermore, $\D_C(G)=\{A \subseteq C \mid \forall x \in A,\ \priv(A,x)\cap S\neq\emptyset\}$ and
  \begin{enumerate}
      \item $\D_C(G)$ and $\D(G)$ are in bijection,\label{item:k-c-i1}
      \item $(C,\D_C(G))$ is an independence system.\label{item:k-c-i2}
  \end{enumerate}
\end{proposition}

In the following, we consider a split graph $G$ and a neighborhood inclusion poset $P_G$.
As $P_G$ is a neighborhood inclusion poset, Equality~\eqref{eqn:main} applies.
The next proposition allows us to consider a decomposition of $G$ into a maximal independent set $S$, and a clique $C$, such that $S\subseteq \Min(P_G)$.

\begin{proposition}\label{prop:hypothesis-SminP}
  Let $G$ be a split graph and $P_G$ be a neighborhood inclusion poset on~$G$.
  Then there exists a decomposition of $G$ into a maximal independent set $S$, and a clique $C$, such that $S\subseteq \Min(P_G)$.
\end{proposition}

\begin{proof}
  Let $S,C$ be a decomposition of $G$ that maximizes the independent set $S$.
  If $x\in S$ and $x\not\in \Min(P_G)$, then there exists some $y_x\in C$ such that $y_x\leq x$, $N[x]=N[y_x]$, and thus such that $S\setminus \{x\}\cup\{y_x\}$ and $C\setminus \{y_x\}\cup\{x\}$ is still a decomposition of $G$ that maximizes the independent set.
\end{proof}

We now define 
\[
  \D_C(G,P_G)\eqdef\{D_C \mid D\in \D(G)~\text{and}~\downarrow D\in \ID(G,P_G)\},
\]
and show that Proposition~\ref{prop:kante-characterization} extends for this set.

\begin{lemma}\label{lemma:subindependencesystem}
  Let $G$ be a split graph with maximal independent set $S$ and clique~$C$, and $P_G$ be a neighborhood inclusion poset on~$G$.
  Then $\D_C(G,P_G)$ and $\ID(G,P_G)$ are in bijection, and $\D_C(G,P_G)\subseteq \D_C(G)$.
\end{lemma}

\begin{proof}
  The bijection between $\D_C(G,P_G)$ and $\ID(G,P_G)$ follows from Propositions~\ref{prop:domantichain}, \ref{prop:kante-characterization} and Equality~\eqref{eqn:main}, where to every $A\in \D_C(G,P_G)$ corresponds a unique $I\in \ID(G,P_G)$ such that $I=\downarrow(A\cup (S\setminus N(A)))$, and to every $I\in \ID(G,P_G)$ corresponds a unique $A\in \D_C(G,P_G)$ such that $A=\Max(I)\cap C$.

  The inclusion $\D_C(G,P_G)\subseteq \D_C(G)$ follows from Equality~\eqref{eqn:main}, as $A\in \D_C(G,P_G)$ implies $A=D_C$ for some $D\in \D(G)$ such that $\downarrow D\in \D(G,P_G)$.
\end{proof}

\begin{lemma}\label{lemma:split-IS}
  Let $G$ be a split graph with maximal independent set $S\subseteq \Min(P_G)$ and clique~$C$, and $P_G$ be a neighborhood inclusion poset on~$G$.
  Then $(C,\D_C(G,P_G))$ is an independence system that can be enumerated with polynomial delay given $G$ and $P_G$.
\end{lemma}

\begin{proof}
  We first show that $(C,\D_C(G,P_G))$ is an independence system, by proving that if $A\in \D_C(G,P_G)$ and $A$ is not empty, then removing any element in $A$ yields another set in $\D_C(G,P_G)$.
  Let $\emptyset\neq A\subseteq C$ such that $A\in \D_C(G,P_G)$.
  Let us assume toward a contradiction that there exists $x\in A$ such that $A\setminus\{x\}\not\in \D_C(G,P_G)$.
  By~Proposition~\ref{lemma:subindependencesystem}, since $A\in \D_C(G,P_G)$ and
  since $\D_C(G)$ is an independence system, both $A$ and $A\setminus \{x\}$ belong to $\D_C(G)$.
  Let $D,D'\in \D(G)$ such that $A=D_C$ and $A\setminus\{x\}=D'_C$.
  By Proposition~\ref{prop:kante-characterization}, $D'=D\setminus \{x\}\cup \{s_1,\dots,s_k\}$ where $\{s_1,\dots,s_k\}=\priv(A,x)\cap S$.
  As by hypothesis $A\setminus\{x\}\not\in \D_C(G,P_G)$, there exists $D^*\in \D(G)$ such that $\downarrow D^*\subsetneq \downarrow D'$.
  Now, note that $\Min(P_G)\cap D'\subseteq D^*$, as otherwise there exists $w\in \Min(P_G)\cap D'\setminus D^*$, hence $D^*\subseteq \downarrow (D'\setminus \{w\})$,
  and we deduce that $\downarrow (D'\setminus \{w\})$ dominates $G$. But then by Proposition~\ref{prop:domantichain}, $\Max(\downarrow (D'\setminus \{w\}))=D'\setminus \{w\}$ dominates $G$, which contradicts the fact that $D'$ is a minimal dominating set.
  Therefore $\{s_1,\dots,s_k\}\subseteq D^*$ and as $\downarrow D^*\subsetneq \downarrow D'$, there exist $u\in D^*$ and $v\in D'\setminus \Min(P_G)\setminus D^*$ such that $u<v$.
  Note that $v\in D$ (as $D'\subseteq D$) and $v\neq x$ (as $x\not\in D'$).
  Let $D^\circ=D^*\cup \{x\}\setminus\{s_1,\dots,s_k\}$. 
  Clearly $D^\circ$ dominates $G$.
  As $\downarrow D^*\subseteq \downarrow D'$, $\downarrow D^\circ \subseteq \downarrow D$.
  Moreover $\downarrow D^\circ\subsetneq \downarrow D$ as $v\in D$ and $v\not\in D^\circ$.
  This contradict the hypothesis that $A\in \D_C(G,P_G)$.
  Hence $A\cup \{x\}\in \D_C(G,P_G)$.

  Now, note that testing whether some arbitrary set $A\subseteq C$ belongs to $\D_C(G,P_G)$ can be done in polynomial time in the sizes of $G$ and $P_G$:
  first compute the unique $D\in \D(G)$ such that {$D_C=A$}, using Proposition~\ref{prop:kante-characterization}, and test whether $\downarrow D\in \ID(G,P_G)$ by checking if ${\priv(\downarrow D, x)}\neq \emptyset$ for every $x\in D$.
  Hence, $\D_C(G,P_G)$ can be enumerated with polynomial delay by adding vertices of $C$ one by one from the empty set to maximal elements of $\D_C(G,P_G)$, checking at each step whether the new set belongs to $\D_C(G,P_G)$.
  Repetitions are avoided with a linear order on vertices of $C$; see~\cite{kante2014enumeration} for further details on the enumeration of an independence system.
\end{proof}

We conclude to a polynomial delay algorithm to enumerate $\ID(G,P_G)$ whenever $G$ is a split graph and $P_G$ is a neighborhood inclusion poset on $G$.
The algorithm first computes a decomposition $S,C$ that maximizes the independent set, makes $S$ a subset of $\Min(P_G)$ using Proposition~\ref{prop:hypothesis-SminP}, and enumerates the independence system $(C,\D_C(G,P_G))$ with polynomial delay using Lemma~\ref{lemma:split-IS}.
For every $A\in \D_C(G,P_G)$, it outputs the unique corresponding $I=\downarrow D$ such that ${D_C=A}$ using Lemma~\ref{lemma:subindependencesystem}.
This can clearly be done with polynomial delay.
We conclude with the following result.

\begin{theorem}\label{thm:split}
  There is a polynomial delay algorithm for \IDomEnum{} whenever $G$ is split and $P_G$ is a neighborhood inclusion poset.
\end{theorem}

\subsection{Triangle-free graphs and weak neighborhood inclusion posets}

In \cite{bonamy2019triangle}, the authors give an output-polynomial algorithm to enumerate minimal dominating sets in triangle-free graphs, i.e., graphs with no induced clique of size three.
These graphs include bipartite graphs.
We rely on this algorithm to show that $\ID(G,P_G)$ can be enumerated in output-polynomial time in the same graph class, whenever $P_G$ is a weak neighborhood inclusion poset on~$G$.
Our argument is based on the next observation.

\begin{proposition}\label{prop:star-in-poset}
  Let $G$ be a triangle-free graph and $P_G$ be a weak neighborhood inclusion poset on~$G$.
  Then $P_G$ is of height at most two, and it is partitioned into an antichain $A$ of isolated elements (that are both minimal and maximal in $P_G$), and a family $\S$ of $k$ disjoint stars%
  \footnote{$S_i$ induces a star in the Hasse diagram of $P_G$.} 
  $S_1,\dots,S_k$ of respective center $u_1,\dots,u_k$ such that either $S_i=\downarrow u_i$ or $S_i=\uparrow u_i$, for all $i\in [k]$.
  Furthermore, vertices in $S_i\setminus \{u_i\}$ are of degree one in $G$.
\end{proposition}

\begin{proof}
  This situation is depicted in Figure~\ref{fig:posetstars}.
  We first show that $P_G$ is of height at most two.
  Suppose that there exist $x,y,z$ such that $x<y<z$.
  Then $xy,xz,yz\in E(G)$ which contradicts the fact $G$ is triangle-free.

  \begin{figure}
    \center
    \includegraphics[scale=\figurescale]{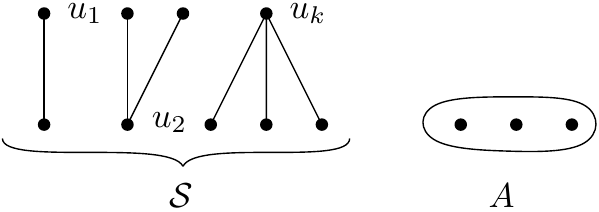}
    \caption{The situation of Proposition~\ref{prop:star-in-poset}.}
    \label{fig:posetstars}
  \end{figure}

  Let us now prove the rest of the proposition.
  Let $A=\Min(P_G)\cap\Max(P_G)$, and $B=P_G\setminus A$.
  Let $S\subseteq B$ be a connected component in the Hasse diagram of $P_G$, and let $x,y\in S$ such that $x<y$.
  Two symmetric cases arise depending on whether $N[x]\subseteq N[y]$ or $N[x]\supseteq N[y]$.
  If $N[x]\subseteq N[y]$ then $x$ is of degree one in $G$ (or else the other neighbor of $x$ would be connected to both $x$ and $y$ and would induce a triangle in $G$).
  Moreover, every other element $z\neq x$ that is comparable with $y$ verifies $N[z]\subseteq N[y]$ (or else it verifies $N[z]\supseteq N[y]$ and $xyz$ induces a triangle in~$G$), hence is of degree one (by previous remark).
  Also, it verifies $z<y$ as $P_G$ is of height at most two.
  Hence $S$ induces a star of center $y$ in the Hasse diagram of $P_G$, such that $S=\downarrow y$, and where every vertex in $S\setminus \{y\}$ is of degree one in $G$.
  The other case $N[x]\supseteq N[y]$ leads to the symmetric situation where $S=\uparrow x$ and where every vertex in $S\setminus \{x\}$ is of degree one in $G$.
\end{proof}

In the following, we denote by $\{v_i^1,\dots,v_i^l\}$ the set of branches of some star $S_i\in \S$, $i\in [k]$, and by $u_i$ its center.
Then, we denote by $G_{re}$ and $P_{G_{re}}$ the {\em reduced graph and poset} obtained from $G$ and $P_G$, where every star $S_i\in \S$ had its branches $\{v_i^1,\dots,v_i^l\}$ contracted into a single element $v_i$, and where every edge $u_iu_j$ that connects two distinct stars $S_i,S_j$ in $G$ has been removed.
We denote by $B_u$ and $B_v$ the sets $B_u=\{u_1,\dots,u_k\}$ and $B_v=\{v_1,\dots,v_k\}$.
The resulting graph is detailed below and is given in Figure~\ref{fig:decomposition1}.
Observe that $P_{G_{re}}$ is partitioned into an antichain $A$ of isolated elements (that are both minimal and maximal in $P_{G_{re}}$, and left untouched by our transformation), and a set $B=B_u\cup B_v=\{u_1,v_1,\dots,u_k,v_k\}$ of $k$ disjoint chains $u_iv_i$ (such that either $u_i<v_i$ or $v_i<u_i$), $i\in [k]$.
The graph $G_{re}$ is partitioned into one triangle-free graph induced by $A$ (left untouched by our transformation), and an induced matching $\{u_1v_1,\dots,u_kv_k\}$ ($B_u$ and $B_v$ induce two independent sets), where $v_i$ is disconnected from $A$, and $u_i$ is arbitrarily connected to $A$, for every $i\in [k]$.
Clearly, $G_{re}$ and $P_{G_{re}}$ can be constructed in polynomial time in the sizes of $G$ and $P_G$.
The following property is implicit in \cite{kante2014enumeration} and can also be found in the Ph.D.~thesis of Mary \cite{mary2013enumeration}.

\begin{figure}
  \center
  \includegraphics[scale=\figurescale]{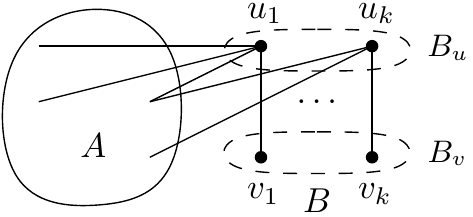}
  \caption{The decomposition $(A,B)$ of a reduced triangle-free graph $G_{re}$.}
  \label{fig:decomposition1}
\end{figure}

\begin{proposition}[\cite{mary2013enumeration,kante2014enumeration}]\label{prop:redundant}
  Let $G$ be a graph and $uv$ be an edge of $G$. 
  Then $\D(G)=\D(G-uv)$ whenever there exists $u'\neq u$, $v'\neq v$ such that $N_{G-uv}[u']\subseteq N_{G-uv}[u]$ and $N_{G-uv}[v']\subseteq N_{G-uv}[v]$.
  Such an edge $uv$ is called redundant.
\end{proposition}

\begin{lemma}\label{lemma:redundant-bij}
  There is a bijection between $\ID(G,P_G)$ and $\ID(G_{re},P_{G_{re}})$.
\end{lemma}

\begin{proof}
  Let $S$ be a star of Proposition~\ref{prop:star-in-poset} of center $u$ and branches $v^1,\dots,v^l$.
  Then, observe that $v^1,\dots,v^l$ are false twins in $G$, i.e., $N(v^i)=N(v^j)=u$ for all $i,j\in [l]$.
  It is easy to see that a minimal dominating set contains $v^i$ for one such $i$ if and only if contains the whole set $\{v^1,\dots,v^l\}$ as a subset.
  Hence, the contraction of all branches $\{v^1,\dots,v^l\}$ of $S$ into a representative vertex $v$ in both $G$ and $P_G$ has no impact on the complexity of enumerating minimal dominating sets: one can replace $v$ by $\{v^1,\dots,v^l\}$ for every $D\in \D(G)$ such that $\downarrow D\in \ID(G,P_G)$ and $v\in D$ to obtain solutions of the graph before contraction.
  As for the deleted edges $u_iu_j$, $i,j\in [k]$, $i\neq j$, they are all redundant as $N_{G-u_iu_j}[v_p]\subseteq N_{G-u_iu_j}[u_p]$ for all $i,j,p\in [k]$, $i\neq j$.
  By Proposition~\ref{prop:redundant}, they can be removed from $G$ with no incidence on domination.
\end{proof}

\begin{proposition}\label{prop:min-Bu}
  For every minimal dominating set $D$ such that $\downarrow D \in \ID(G_{re},P_{G_{re}})$, $\Min(P_{G_{re}})\cap B_u\subseteq D$.
\end{proposition}

\begin{proof}
  Let $u\in B_u\cap \Min(P_G)$ and $v\in B_v$ be the unique vertex such that $u<v$.
  Since $v$ is of degree one in $G$, it must be dominated by either itself, or $u$.
  Since $u<v$, a dominating-ideal that contains $v$ is not minimal. 
  Hence $\Min(P_{G_{re}})\cap B_u\subseteq D$ for all minimal dominating set $D$ such that $\downarrow D \in \ID(G_{re},P_{G_{re}})$.
\end{proof}

Let $G$ be a graph and $W,D$ be two subsets of vertices of $G$.
Recall that $\D_G(W)$ denotes the set of minimal dominating sets of subset $W$ in $G$; see~Section~\ref{sec:preliminaries}.
We now rely on an implicit result from \cite{bonamy2019triangle}, made explicit in~\cite{bonamy2019kt}.

\begin{theorem}[\cite{bonamy2019triangle,bonamy2019kt}]\label{thm:triangle-free}
  There is an algorithm that, given a graph $G$ and a set $W\subseteq V(G)$ such that $G[W]$ is triangle-free, enumerates $\D_G(W)$ in total time $\poly(|G|)\cdot|\D_G(W)|^2$ and polynomial space.
\end{theorem}

Let us define the set $B_w=\Min(B)=\{w_1,\dots,w_k\}$. 
Note that $w_i=\Min_\leq\{u_i,v_i\}$ for all $i\in [k]$.
We now consider the set 
\[
  A'\eqdef{}A\setminus \bigcup_{i=1}^k N[w_i].
\]
Clearly, $G_{re}[A']$ is triangle-free.
Hence, $\D_{G_{re}}(A')$ can be enumerated in output-polynomial time $\poly(|{G_{re}}|)\cdot|\D_{G_{re}}(A')|^2$ using the algorithm of Theorem~\ref{thm:triangle-free}.
We now show how to compute $\ID({G_{re}},P_{G_{re}})$ given $\D_{G_{re}}(A')$.

\begin{figure}
  \center
  \includegraphics[scale=\figurescale,page=2]{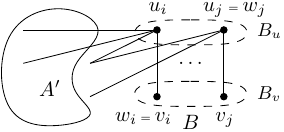}
  \caption{The situation of Lemma~\ref{lemma:triangle-free}.}
  \label{fig:decomposition2}
\end{figure}

\begin{lemma}\label{lemma:triangle-free}
  Le $D$ be a minimal dominating set of $G$.
  Then $\downarrow D\in \ID({G_{re}},P_{G_{re}})$ if and only if $D= D^*\cup \{w_i \mid v_i\not\in N[D^*]\}$, $D^*\in \D_{G_{re}}(A')$.
\end{lemma}

\begin{proof}
  The situation of this lemma is depicted in Figure~\ref{fig:decomposition2}.
  We show the first implication.
  Let $D\in \D(G)$ such that $\downarrow D\in \ID({G_{re}},P_{G_{re}})$, and let $D^*=D\setminus \{w_i \mid w_i \in D\}$.
  Clearly, $D^*$ dominates $A'$.
  Let $t\in D^*$.
  We show that it has a private neighbor in $A'$.
  Let $a$ be a private neighbor of $t$ (w.r.t.~$D$) such that $a\not\in A'$.
  If no such $a$ exists, then we proved our claim, as in that case $t$ must have a private neighbor in $A'$.
  Else, $a$ belongs to $N[w_i]$ for some $i\in [k]$.
  If~$w_i=u_i$ then by Proposition~\ref{prop:min-Bu} $w_i\in D$ which contradicts the fact that $a$ is a private neighbor of $t$.
  If~$w_i=v_i$, then $a\in \{u_i,v_i\}$. 
  Since either $u_i$ or $v_i$ belongs to $D$ (as $v_i$ is of degree one), it must be that either $t=u_i$ or $t=v_i$. 
  As $t\neq w_i=v_i$, we know that $t=u_i$.
  In that case, $t$ has another private neighbor $a'\neq a$ that is non-adjacent to $v_i$ (or else $\downarrow D$ is not a minimal dominating-ideal as $t=u_i$ can be replaced by $v_i$, $a\in N[v_i]$, and $v_i<u_i$).
  At last, if $a'$ belongs to $w_{j}$ for some $j\in [k]$, then $w_{j}=u_{j}$ (as $N[v_i]=\{u_i,v_i\}$ and $B$ is an induced matching, hence $a\neq u_j$) which by Proposition~\ref{prop:min-Bu} is absurd, as $w_j\in D$.
  Hence $a'\in A'$, which proves our claim.
  Hence $D^*$ minimally dominates~$A'$, i.e., $D^*\in \D_{G_{re}}(A')$.
  Now, note that $w_i\in D$ if and only if $v_i \not\in N[D^*]$.
  Indeed, if $v_i \not\in N[D^*]$ then $w_i \in D$ (as otherwise $w_i \not\in D$, by Proposition~\ref{prop:min-Bu} $w_i=v_i$, hence $u_i \in D$, $u_i\in D^*$, and $v_i\in N[D^*]$ which is absurd).
  If $v_i \in N[D^*]$, then $u_i\in D^*$, $w_i=v_i$, and $w_i\not\in D$ or else $\{u_i,v_i\}\subseteq D$ which is absurd since $D$ is an antichain.
  Hence $D= D^*\cup \{w_i \mid v_i\not\in N[D^*]\}$ which concludes the first implication.

  We show the other implication.
  Let $D^*\in \D_{G_{re}}(A')$ and $D=D^*\cup \{w_i \mid v_i\not\in N[D^*]\}$.
  Clearly $D$ dominates ${G_{re}}$ as for all $i\in[k]$, either $v_i\in N[D^*]$ and therefore $u_i\in D^*$ (as $v_i$ is disconnected from $A'$) and $N[w_i]$ is dominated, or $v_i\not\in N[D^*]$ and $w_i$ dominates $N[w_i]$.
  Note that if $t\in D^*$ then it has private neighbors in $A'$ that are not adjacent to any $w_i$ (by construction), hence such that no ideal $I\subsetneq \downarrow (D\setminus \{t\})$ can dominate.
  If $t\in D\setminus D^*$ then $t=w_i$ for some $i\in [k]$, it has $v_i$ for private neighbor, and it is minimal in $P_{G_{re}}$.
  Hence $\downarrow D$ is minimal dominating-ideal of $G$.
\end{proof}

We conclude to the existence of an output-polynomial algorithm to enumerate the set $\ID(G,P_G)$ whenever $G$ is triangle-free and $P_G$ is a weak neighborhood inclusion poset.
The algorithm first computes $G_{re}$ and $P_{G_{re}}$ in polynomial time in the sizes of $G$ and $P_G$, and then enumerates $\ID(G,P_G)$ using Lemmas~\ref{lemma:redundant-bij} and~\ref{lemma:triangle-free}.

\begin{theorem}\label{thm:bipartite}
  There is an algorithm that, given a triangle-free graph~$G$ and a weak neighborhood inclusion poset $P_G$, enumerates $\ID(G,P_G)$ in output-polynomial time.
\end{theorem}

We note that as $G_{re}[A]$ can yield any triangle-free graph in our construction, improving the algorithm of Theorem~\ref{thm:bipartite} to run with polynomial delay constitutes a challenging open question~\cite{bonamy2019triangle}.

\section{Conclusion}\label{sec:conclusion}

In this paper, we generalized the two problems of enumerating the minimal transversals of a hypergraph, and the minimal dominating sets of a graph, to the enumeration of the minimal ideals of a poset with the desired property, i.e., transversality and domination.
We showed that the obtained problems are equivalent to the dualization in distributive lattices, even when considering various combined restrictions on graph classes and poset types, including bipartite, split, and co-bipartite graphs, and variants of neighborhood inclusion posets; see Theorems~\ref{thm:maintrans} and~\ref{thm:maindom}.
This study allowed us to consider the complexity of the problem under new parameters.
For combined restrictions that are not considered in Theorem \ref{thm:maindom}, we showed that the problem is tractable relying on existing algorithms from the literature; see Theorems~\ref{thm:split} and~\ref{thm:bipartite}.
A summary of the obtained complexities is given in Figure~\ref{fig:sum}.

\begin{figure}
  \small
  \center
  \begin{tabular}{ | C{3cm} | C{2.3cm} | C{2.3cm} | C{2.3cm} |N} 

    \hline 
    Graph classes & N.I. posets & Weak N.I. posets & Arbitrary posets &\\
    \hline
    \hline 
    Bipartite & {\sf OutputP} & {\sf OutputP} &  \textsc{D}-hard \\
    \hline 
    Split & {\sf PolyD} &  \textsc{D}-hard &  \textsc{D}-hard \\
    \hline 
    Co-bipartite &  \textsc{D}-hard &  \textsc{D}-hard &  \textsc{D}-hard \\
    \hline 
  \end{tabular} 
  \caption{Summary of the complexity results obtained in Theorems~\ref{thm:maindom},~\ref{thm:split} and~\ref{thm:bipartite} under combined restrictions on graph classes and poset types. 
  {\sf OutputP} stands for output-polynomial, and {\sf PolyD} for polynomial delay.
  N.I.~stands for neighborhood inclusion, and \textsc{D}-hard for \DualEnum{}-hard.}\label{fig:sum}
\end{figure}

We leave open the complexity status of distributive lattice dualization in general.
We point that the results of Theorems~\ref{thm:split} and \ref{thm:triangle-free} characterize couples of antichains (coded by the graph) and distributive lattices (coded by the poset) for which the dualization is tractable.
For future work, we would be interested in characterizations that only depend on the poset, in order to obtain classes of lattices for which the dualization is tractable, as in \cite{defrain2019dualization,elbassioni2009algorithms}, using graph structures presented in this paper.

\bibliographystyle{alpha}
\bibliography{main}

\end{document}